\documentclass[conference]{IEEEtran}
\IEEEoverridecommandlockouts
\usepackage{cite}
\usepackage{amsmath,amssymb,amsfonts}
\usepackage{algorithmicx}
\usepackage{graphicx}
\usepackage{textcomp}
\usepackage{xcolor}
\def\BibTeX{{\rm B\kern-.05em{\sc i\kern-.025em b}\kern-.08em
    T\kern-.1667em\lower.7ex\hbox{E}\kern-.125emX}}
    
\usepackage{booktabs} 
\usepackage{algorithm}
\usepackage[noend]{algpseudocode}
\usepackage{listings}

\usepackage{balance}

\usepackage[utf8]{inputenc}
\usepackage{tikz} 
\usepackage[T1]{fontenc}
\usetikzlibrary{%
  arrows,%
  calc,
  shapes,
  arrows,
  shapes.misc,
  shapes.arrows,
  chains,%
  matrix,%
  positioning,
  scopes,%
  decorations.pathmorphing,
  shadows%
}
\usepackage{amsmath}
\usepackage{relsize}  

\usepackage{amsthm}

\newtheorem{lemma}{Lemma}

\begin{document}

\title{Read-Uncommitted Transactions for \\ Smart Contract Performance
 \thanks{This work was funded by the National Science Foundation (NSF) under Grant Numbers 1717515 and 1740095.}
}

\author{\IEEEauthorblockN{Victor Cook\IEEEauthorrefmark{1},
Zachary Painter\IEEEauthorrefmark{2},
Christina Peterson\IEEEauthorrefmark{3},
Damian Dechev \IEEEauthorrefmark{4}}\\
\IEEEauthorblockA{\textit{Computer Science Department},
\textit{University of Central Florida}\\
Orlando, USA \\
\IEEEauthorrefmark{1}victor.cook@knights.ucf.edu,
\IEEEauthorrefmark{2}zacharypainter@knights.ucf.edu,\\
\IEEEauthorrefmark{3}clp8199@knights.ucf.edu,
\IEEEauthorrefmark{4}dechev@cs.ucf.edu}
}

\maketitle

\begin{abstract}
Smart contract transactions demonstrate issues of performance and correctness that application programmers must work around.  Although the blockchain consensus mechanism approaches ACID compliance, use cases that rely on frequent state changes are impractical due to the block publishing interval of $O(10^1)$ seconds.  The effective isolation level is \textsc{Read-Committed}, only revealing state transitions at the end of the block interval. Values read may be stale and not match program order, causing many transactions to fail when a block is committed.  
This paper perceives the blockchain as a transactional data structure, using this analogy in the development of a new algorithm, \textit{Hash-Mark-Set} (HMS), that improves transaction throughput by providing a \textsc{Read-Uncommitted} view of state variables.  HMS creates a directed acyclic graph (DAG) from the pending transaction pool. The transaction order derived from the DAG is used to provide a \textsc{Read-Uncommitted} view of the data for new transactions, which enter the DAG as they are received.
An implementation of HMS is provided, interoperable with Ethereum and ready for use in smart contracts.  Over a wide range of transaction mixes, HMS is demonstrated to improve throughput.  A side product of the implementation is a new technique, \textit{Runtime Argument Augmentation} (RAA), that allows smart contracts to communicate with external data services before submitting a transaction.  RAA has use cases beyond HMS and can serve as a lightweight replacement for blockchain oracles. 
\end{abstract}

\begin{IEEEkeywords}
Blockchain, Smart contracts, Concurrent algorithms, Transaction throughput
\end{IEEEkeywords}

\section{Introduction} 
Blockchains rely on a consensus mechanism to agree upon the sequencing of client transactions in a block, committing transactions as a group to the distributed ledger. \textit{Smart contracts} are the interface to process client requests and send transactions to the blockchain peer network.  A block of transactions must be validated to ensure that the sequence is consistent.  All peers on the network perform the validation step by re-executing the transactions within the block and checking that the initial and final states match, introducing latency.  

Latency resulting from the publishing and validation of a block decreases the success rate for the transactions in the block due to the possibility of stale reads of state variables, also known as \textit{storage variables}. Changes to storage variables are only visible after they are committed to a published block.  This isolation level of intra-block transactions is called \textsc{Read-Committed}.  Transactional reads of storage variables can become outdated while waiting on the validation step since other published blocks may update the storage variables, leading to transaction failure.  Additionally, since read operations can only access the published storage variable value, intra-block changes can also cause a transaction to fail due to a stale read. 

A smart contract transaction is a concurrent method, often with semantic dependencies. The way block publishing commits multiple smart contract transactions simultaneously is analogous to the way a transactional data structure~\cite{herlihy_koskinen_2008,zhang_dechev_2016} commits multiple concurrent methods in what appears to be a single atomic step.  Using this analogy, the blockchain is a blind transactional data structure that selects and sequences concurrent method calls without regard for their semantics, causing many to fail due to the restrictive \textsc{Read-Committed} isolation level. An ideal algorithm for blockchain transactions would consider transaction semantics and include all related transactions as a series in a block commit.  

In this paper, we present \textit{Hash-Mark-Set} (HMS), an algorithm that increases the throughput of smart contract transactions by providing a \textsc{Read-Uncommitted} view of the storage variables. HMS organizes the pool of pending transactions (TxPool) on specific storage variables in a directed acyclic graph (DAG) that establishes an ordering among the transactions and enables an uncommitted view of the storage variables to be retrieved.  HMS reduces transactional failures because the \textsc{Read-Uncommitted} view increases the likelihood that a transaction has consistent inputs. Latency is also reduced because concurrent actors will no longer need to wait until a block is committed to see a change in storage variables that is likely to be committed in the next block or two. We integrate HMS into smart contracts through \textit{Runtime Argument Augmentation} (RAA), our proposed technique that allows smart contracts to communicate with external data services prior to sending a transaction. 

An interoperable implementation of the solution is provided and demonstrated on the Ethereum blockchain.  State throughput, defined in Section~\ref{SubSection:StateThroughput} as throughput of \textit{successful} blockchain transactions, increases by a factor of five across a range of transaction mixes.  By adding the cooperation of blockchain miners, state throughput increases by an order of magnitude.

We make the following contributions:
\begin{enumerate}
\item Present \textit{Hash-Mark-Set} (HMS), an algorithm organizing a pool of pending transactions that share state variables.
\item Introduce \textit{Runtime Argument Augmentation} (RAA), a novel technique for smart contracts to communicate with external data services.
\item Demonstrate improvements to latency and state throughput when HMS provides a \textsc{Read-Uncommitted} view to smart contract clients. 
\item Provide an interoperable implementation for Ethereum:\\ https://github.com/area67/sereth
\end{enumerate}

\section{Background}

A few issues specific to the blockchain are discussed in this section.

\subsection{Blockchain Transactions}
\label{Subsection:Blockchain}
A \textit{blockchain} is a distributed ledger maintained by one or more peers following a communication protocol and agreeing on a consensus mechanism.  The ledger is written in chunks called \textit{blocks} that are linked in a practically unforgeable cryptographic \textit{chain}, replicated among many peers to avoid dependence upon a single entity. 
State variables that are recorded on the blockchain are called \textit{storage} variables in Ethereum.  
Bitcoin was the first blockchain \cite{Nakamoto2008-wx}, providing transactions on a store and exchange of value, \textit{i.e.} a currency.  A transaction is a concurrent method call that if successful, changes the state of the ledger. A block may contain from zero up to a finite number of transactions, typically $O(10^1)$ to $O(10^3)$.  Well known database transaction models such as ACID and BASE are applicable to the blockchain\cite{Tai2017-bt}, motivating our use of isolation levels \textsc{Read-Uncommitted} and \textsc{Read-Committed}.   

\subsection{Concurrent Smart Contracts}
Going beyond exchange of value, later blockchains added the ability to program arbitrary instruction sequences in a transaction.   Their programming languages are Turing complete \cite{Buterin2013-ej} and their programs are called \textit{smart contracts} \cite{Szabo1997-gi}.  Concurrency is framed in the words of Sergey and Hobor, ``Accounts using smart contracts in a blockchain are like threads using concurrent objects in shared memory''  \cite{Sergey2017-zx}.  Herlihy endorsed this line of reasoning in a keynote address \cite{Herlihy2017-av}, exhorting concurrency researchers to ``civilize'' the blockchain.  
    
Invoking a smart contract function that \textit{may} change ledger state creates a transaction and sends it to the network. 
It should be noted that some smart contract functions, designated \textit{pure} or \textit{view}, cannot change ledger state and they do not create transactions.
The unprocessed transaction pool of pending transactions is referred to as the \textit{TxPool}.
The network of peers is a concurrent system and it follows that its incoming transactions, found in the TxPool, is a concurrent history.  
The \textit{real time ordering} of a concurrent history is a total ordering over the transactions in a concurrent history such that transaction $T_1$ is ordered before transaction $T_2$ if $T_1$ is received by the TxPool before $T_2$.

\subsection{Miner Privilege}

On Ethereum, Bitcoin and many other blockchains, the inclusion and sequencing of transactions in a block does not follow real time order,
rather transactions are arranged in a total order that is arbitrary and subject to the same economic incentives that drive blockchain progress \cite{Wang2018-ar}.  This is called the \textit{block order}. Special peers, called \textit{miners} have the privilege of deciding what goes into a block and in what order.  Each transaction is isolated and a miner generally has no way of knowing if one may depend upon another, so the rules for selecting transactions are flexible.  Miners generally favor transactions with higher fees, but they may favor some peers, including themselves.  They may use altruistic criteria such as including only small transactions or those from peers with low bandwidth.  

The discretion given to miners in the protocol works as if the scheduler of a CPU could favor particular threads.  Ethereum miners read the TxPool grouped by peer addresses (aka threads) with transactions ordered by a counter called a \textit{nonce}.  Miners may favor an address and include its transactions before another peer without regard for the real time order in which they were received.  Miners may refuse to include any transactions sent from particular addresses.  But a miner may not commit a transaction from a given address to a block out of nonce order. This means that blockchain transactions from the same address are executed in the order they are sent, while the order of transactions from different addresses is not defined.
Since a blockchain transaction is a concurrent method, we can describe this behavior as being equivalent to \textit{sequential consistency}, a correctness property such that history of methods is equivalent to a legal sequential history, and all methods take effect in program order~\cite{herlihy2011art}.

The TxPool is shared by peers on the network, including miners.  Intuitively, if communication were instantaneous, all peers would see the same TxPool, and the order in which the transactions were received would match their real time ordering, i.e. the order in which they were sent.  Miner privilege would still allow the transactions to be placed in a block in an order different from the real time order. 
The outcome of our \textsc{Read-Uncommmitted} view of state is subject to network synchronization and miner privilege.  Information about the TxPool is not available to the smart contract as it submits transactions.  

\subsection{Block Publishing and Validation}
\label{Subsection:BlockPublishing}
Blocks of selected transactions are committed all at once in a super transaction called \textit{block publishing}.
Transactions are interpreted sequentially within a block according to the block order, using the previous ledger state (block) as the initial context.
Changes to storage variables are not visible until they are committed to a block and the block is published.
The changes to storage variables during the interval of block publication are called \textit{intra-block} changes.  
Transactions within the block are affected by the intra-block changes, but post-publication transactions read the block final value of a storage variable from the previous block, none of the current changes. 
Once published there is no opportunity to re-order the concurrent methods. 
These values were read from the previous block, published \textit{block interval} seconds ago. 
The block interval defines the latency.  

Block publishing is effectively a read lock until the next block is committed.  
Dirty reads are not allowed. 
In database terms the isolation level of intra-block transactions is \textsc{Read-Committed}.  
To accept a published block every peer must perform \textit{block validation}, the task of checking that the block is consistent with the state of the network. 
Transactions committed to a block must be consistent in that they must include the effects of all previous transactions.
The process of peers redundantly validating transactions in a block is called transaction \textit{replay}.  
Block publishing and validation takes a significant amount of time $O(10^1)$ to $O(10^2)$ seconds, creating latency.

Since only the final state of the block is published, intermediate states become invisible without a detailed replay of the transactions in the block, something that a typical smart contract cannot do.  
The loss of intermediate states during a block update is a consequence of the \textsc{Read-Committed} view of state variables.
This low isolation level avoids blocking but may allow a great number of transactions to be rejected later as inconsistent.  
Transactions that seem valid when submitted are rejected because the values on which they are based were stale.
The number of transactions that are rejected impacts state throughput. 
Where state changes are frequent and there are many transactions in the pool to be interleaved in a block, a large percentage of transactions fail.  
To say a transaction failed means that it would have violated the consistency of the sequential history of the block in which it is embedded. 
To keep the sequential history of the block consistent, the transaction is included in the block, but has no effect on the system state. 
In database terms the transaction was \textit{rolled back}.  
A principle cause of failure is the high latency imposed on reading changes to persistent storage variables.  

\subsection{Blockchain Oracles}
A characteristic of the blockchain is that security concerns related to the adversarial distributed environment impose restrictions on information transfer.  Unassisted, smart contracts operate in a bubble, allowed to view only public blockchain state variables via getter functions and not allowed to call any outside sources of information.  The discussion in Section~\ref{Subsection:BlockPublishing} about peers replaying blocks can explain this.  Since all peers must replay and validate the block, they all must see the same state changes.  If a contract is using an outside source of information, no matter how reliable, it may change with time or due to corrections or it may become unavailable.  This would cause some peers to see a different state than others, and the block could not be validated.  The problem can be solved with a smart contract that mediates a secure and verifiable connection to external data feeds \cite{Zhang2016-pq}. Such a service is also called a \textit{blockchain oracle} \cite{Xu2016-mm}, \cite{Bartoletti2017-vq}.

\subsection{Challenging Use Cases}
\label{SubSection:UseCases}
Blockchain performance, measured in terms of transaction throughput and latency, is a limiting factor for many use cases \cite{Swan2015-re, Gatteschi2018-uf, Staples2017-ri, Vukolic2016-ls}.   Latency and throughput are considered together in this paper because the \textsc{Read-Committed} latency of state variable limits the throughput of successful transactions.  This ubiquitous blockchain latency has been dubbed, `the long system freeze'' \cite{Eyal2016-fr}.  Our example use case is a decentralized market to buy and sell assets, a core use case driving blockchain research and investment.  This example also represents the general case of concurrent actors reading a time sensitive shared state variable.

Say that trading opens at a certain price, visible to all buyers.  Orders are received on the network to be processed.  To simplify, orders must be at the exact price, \textit{i.e.} there are no limit or market orders.  The price changes frequently and unpredictably due to market dynamics.  If 100 orders are received at the published price near the start of a block interval and the price changes after the first order, then only one will be accepted.  Blockchain correctness (safety, consistency) is preserved by the expedient of invalidating 99 of the 100 transactions in this example, clearly an inefficient mechanism.  

Due to miner privilege, the first order submitted in time may not be the first included in the block. Progress of the system cannot be fair in any case because there is not enough information in the TxPool on which to base a real time order of the requests from different peers.  Even with such information, miners are not bound to prevent starvation, quite the contrary they may cause it.  Information is also hidden from the buyers querying the smart contract for the price.  Block replay is not available within the smart contract.  Unless it is separately analyzed, 98 of the 99 price changes are invisible to participants and valuable market information about intermediate price changes is lost.  The arbitrary transaction priority combined with read latency also creates a vulnerability known as \textit{blockchain frontrunning} \cite{Swende2017-px}.

\section{Methodology}
This paper presents Hash-Mark-Set (HMS), an algorithm that overcomes the limitations of the \textsc{Read-Committed} isolation level by providing a \textsc{Read-Uncommitted} view of storage variables.  The \textsc{Read-Uncommitted} view alleviates the problems in the example of Section~\ref{SubSection:UseCases}.  Clients can observe partial changes within the block prior to publishing, reducing the chance that a transaction will fail due to a stale read.  The \textit{Mark} in HMS also establishes a partial intra-block order that a cooperating miner can enforce.
Such cooperation is reasonable given financial incentives that might be offered by decentralized asset exchanges.  

HMS provides a \textsc{Read-Uncommitted} view by maintaining the transactions in a directed acyclic graph (DAG) that represents an ordering among the transactions in the unprocessed transaction pool, \textit{TxPool}, and applying a topological sort to the longest branch to retrieve the value of an unpublished storage variable. To enable the \textsc{Read-Uncommitted} view to be accessible through smart contracts, we propose Runtime Argument Augmentation (RAA), our proposed technique that modifies the Ethereum Virtual Machine (EVM) interpreter to apply the HMS algorithm and access the value of an unpublished storage variable. The RAA technique is made available to users through our proposed smart contract \textit{Sereth}.  

To evaluate the performance benefits of our proposed methodology, we present a new metric, \textit{state throughput}, which measures the throughput of successful transactions. State throughput disregards failed transactions in the throughput measurement, which provides a better representation of the rate at which state changes are made in comparison to raw throughput. In the following subsections, we define state throughput, provide the Sereth smart contract application programming interface, and explain HMS and RAA, the two innovations of this paper.

\subsection{State Throughput}
\label{SubSection:StateThroughput}

Blockchains are different from databases in the following way: failed transactions are included in the persistent shared ledger.  Because a block may include a large percentage of failed transactions, raw throughput of transactions per second is not an adequate measure of performance.  In the example described in Section~\ref{SubSection:UseCases}, raw throughput was 100 per interval, but 99 of 100 transactions fail.   In a database these rolled back transactions would not count in throughput, but in a blockchain they are included in the block.   A new metric, \textit{state throughput}, $T_{state}$, is defined here as the product of the raw throughput and the ratio of transactions included in a block that successfully make state changes.  State throughput divided by raw throughput yields the transaction efficiency $\eta$. 
\begin{equation}
\frac{T_{state}}{T_{raw}} = \eta 
\end{equation}

Transactions in the TxPool form a concurrent history, with a non-deterministic outcome.  We observed that transaction failure can be reduced by obtaining a view of state that is more likely to be consistent at the moment the transaction is committed to a block.  To maximize $\eta$, transactions are organized to provide a predictive view of state, ordering transactions such that the order closely matches the real time order in which the transactions were received. 

\subsection{Sereth Smart Contract}
Our implementation of HMS for Ethereum is called Sereth, a variation of Geth, the name of the standard client.  Sereth is implemented as an interoperable Ethereum client that can be substituted for one or more peers in any standard Ethereum network, public or private.
The Sereth smart contract shown in Listing~\ref{Listing:SerethContract} manages the price and accepts the $set$ and $buy$ transactions from addresses on the blockchain.  The $mark$ and $get$ functions are read only.   They do not create transactions but are used to return the intra-block state that will be used in $set$ and $buy$.  This intra-block state view uses RAA to get the results of the HMS algorithm.  The values are written into the function arguments using RAA and then returned to the calling address.

\begin{lstlisting}[caption={Sereth smart contract.},label={Listing:SerethContract},captionpos=t,float, basicstyle=\fontsize{7}{7}\selectfont\ttfamily]

pragma solidity ^0.4.24;

contract Sereth {
...
// Mark, Set and Get are methods on state variables
// managed by the Hash-Mark-Set algorithm.

    function mark(bytes32[3] raa) 
            private pure returns(bytes32) {
        return raa[1];
    }
    
    function set(bytes32[3] fpv) public {
        // If mark is valid, set new mark and value.
        if (keccak256(fpv[1]) == keccak256(p[1])) {   
            nSet++;
            p[0] = bytes32(msg.sender);
            p[1] = keccak256(fpv[1], fpv[2]);
            p[2] = fpv[2];
        }
    }
    
    function get(bytes32[3] raa) 
            public pure returns(bytes32) {
        return raa[2];
    }

// Function buy() demonstrates a dynamic pricing use case 
// for the Hash-Mark-Set transactional data structure.

    function buy(bytes32[3] offer) public {
        // If mark and price match then buy() succeeds.
        if ((keccak256(offer[1]) == keccak256(p[1])) &&
            (keccak256(offer[2]) == keccak256(p[2]))) {
            nBuy++;
            p[0] = bytes32(msg.sender);
        }
    }
}
\end{lstlisting}

\subsection{Hash-Mark-Set}

Hash-Mark-Set takes advantage of an underutilized communication channel among the peers on a blockchain, the transaction pool (TxPool).  We created a smart contract, \verb+Sereth.sol+, to manage the state variables.  In Sereth, function arguments are formatted so they contain three key elements within the transaction, $address$, $mark$, and $value$. The $address$ field contains the address of the sender of the transaction. The $mark$ field contains a Keccak256 hash \cite{Patil2015-hs} which solidifies a transactions place in a series of Sereth transactions. The $value$ field indicates how the sender would like to modify the state variable. Together, these elements are referred to as a transaction's $AMV$. To create a transaction using the Sereth contract, one must pass in three parameters: $flag$, $previous\_mark$, and $value$. These parameters are referred to as the $FPV$. The $FPV$ is easily visible as a string of bytes within the transactions $input$ field.   

We define a transaction's $mark$ such that given $Txn_1$ which follows $Txn_0$, $Txn_1.mark = Keccak256(Txn_0.mark, Txn_1.val)$. This creates a sequentially consistent ordering between any number of transactions in what we call a $series$. To create a $series$, the $FPV$ of each transaction in the TxPool is extracted from their respective $Data$ fields. By matching the $previous\_mark$ of a transaction with the $mark$ of a different transaction, we can determine a strict order of all Sereth transactions in the current TxPool.   This provides the smart contract with a \textit{Read-Uncommitted} view of the intra-block state.  In addition, because every state change is linked by a unique hash that includes the value, multiple state changes sequenced in the atomic block update are preserved. 

Algorithm~\ref{alg:sereth} shows the HMS algorithm as implemented on the Ethereum blockchain. Users interact with the algorithm through an Ethereum contract. We refer to line $x$ of algorithm $A$ as $A$:$x$.

\begin{algorithm}
\caption{Transaction Serialization Algorithm} \label{alg:sereth}
\begin{algorithmic}[1]
\Procedure{HashMarkSet(input)}{} \Comment{Serialize a blockchain transaction pool}
\State $RAA \gets input$
\State $txnList \gets$ \textsc{Process}(TxPool) \label{l:txnList} \Comment{Filter TxPool}
\If{len($txnList$) $== 0 $} \label{sereth:listLength}
	\State $RAA \gets specialValue$
	\State \textbf{return}
\EndIf
\State $series \gets$ \textsc{Series}($txnList$) \Comment{Create series}
\State $RAA \gets$ \textsc{copy}($series.tail.FPV$) \label{l:RAA}
\EndProcedure
\end{algorithmic}
\end{algorithm}

\begin{algorithm}
	\caption{Process Transactions}\label{alg:serethParse}
	\begin{algorithmic}[1]
		\Procedure{Process(TxPool, input)}{}\Comment{Filter TxPool for HMS transactions}
		\State $filteredList[]$ 
		\For{$txn \in TxPool$}
			\If{\textsc{Signature}($txn$)$~== ``set"$ \& \textsc{Success}($txn$)}
            	\State $txn.FPV \gets txn.input$
				\State $txn.mark \gets$ 
				\Statex \hspace{8em} Keccak($txn.FPV[1], txn.FPV[2]$) 
				\State $filteredList.push$(\textbf{new} Node($txn$))
			\EndIf
		\EndFor
		\State \textbf{return} $filteredList$
		\EndProcedure
		
		\Procedure{Success(txn)}{}\Comment{Determines if a transaction succeeded or not}
		\State $FPV \gets txn.input$
		\If{$FPV[0] == successFlag$ || $FPV[0] == headFlag$}
			\State \textbf{return} $true$
		\EndIf
		\State \textbf{return} $false$
		\EndProcedure
	\end{algorithmic}
\end{algorithm}

A call to HashMarkSet() is made from the EVM interpreter when the transaction being processed has a function signature that matches that of a Sereth transaction. The $RAA$ variable on line~\ref{alg:sereth}:\ref{l:RAA} represents the storage variable value obtained using the RAA technique. We first extract the $RAA$ from the given $input$ field of the transaction we are processing. This process is simple, as each element is stored in a contiguous 32 bytes within $input$. By writing the result of HashMarkSet() to $RAA$, the result will be made visible within the contract's execution.

Algorithm \ref{alg:serethParse} details how the current transaction pool is filtered and then returned to the main function for handling. 

For each transaction in the pool, we check that the function signature is equal to one of the write functions from our HMS contract. Additionally, we check the first 32 bytes of the FPV for a flag indicating one of several possible states for the transaction. Due to this filtering only a small percentage of the TxPool requires processing, so the overhead of HMS is relatively small.

First, the transaction may be one of the first HMS transactions that appeared during the current block. In this case, we consider the transaction a $head~candidate$, meaning that it or another transaction with the same flag will serve as the head of the serialized list of transactions for the current block. This allows us to easily continue the list from the previous block without being able to view the state variable. The second possible state indicates that the transaction is not a head candidate, and at the time of the transaction's submission, it was found to be the successor to the current tail of the series. If a transaction contains neither of these flags, it is considered rejected and is not included in the list of relevant transactions. If a transaction is accepted, The $FPV$ is then extracted from the $input$ field. The $FPV$ contains $previous\_mark$ and $value$, which are the two values needed to calculate the $mark$ of a transaction and determine its place in the series.  A node is created from the transaction for later inclusion in a linked data structure.

Once $txnList$ has been populated by transactions from the TxPool on line~\ref{alg:sereth}:\ref{l:txnList}, we check on line \ref{alg:sereth}:\ref{sereth:listLength} if the list is empty. If so, the submitted transaction is the first Sereth transaction sent in the current block, and the way to know if it matches the previous mark is to check the state variable within the contract. In this case, a flag is written to the $data$ field, which will be visible to the contract. The contract value will be written in the last 32 bytes of the transaction FPV by the sender. 

If the list contains one or more transactions, then we know that there already exists at least one series for the current block. Algorithm \ref{alg:serethSeries} contains the functions which return the most valid series from a list of Sereth transactions. 

\begin{algorithm}
	\caption{Create a Series}\label{alg:serethSeries}
	\begin{algorithmic}[1]
		\Procedure{Series(txnList)}{}\Comment{Create a serialized list from a set of transactions}\label{proc:series}
		\For{$txn \in txnList$} \label{l:nestedforloop}
			\For{$txn2 \in txnList$}
				\If{$txn.mark == txn2.FPV[1]$} \label{l:ifstatementseries}
					\State $txn2.prevTxn \gets txn$
					\State $txn.nextTxns.push(txn2)$
				\EndIf
			\EndFor
		\EndFor
		
		\State $highestDepth \gets 0$
		\State $longestSeries \gets nil$
		\For{$txn \in txnList $\ such that\ $ txn.FPV[0] == headFlag$} \label{series:head}
			\State $depth \gets 1$
			\State $path \gets [txn]$
			\State $maxDepth \gets 0$
			\State $maxPath \gets []$
			\State \textsc{DeepestBranch}($txn, depth, \&maxDepth,$ \Statex \hspace{15em} $path,$ $maxPath$)
			\If{$maxDepth > highestDepth$}
				\State $highestDepth \gets maxDepth$
				\State $longestSeries \gets maxPath$
			\EndIf
		\EndFor
		
		\State \textbf{return} $longestSeries$
		\EndProcedure
		
		\Procedure{DeepestBranch(head, depth, path, maxDepth, maxPath)}{}\Comment{Recursively find deepest branch}
		\If{len($head.nextTxns$)$~== 0$}
			\If{$depth$ > $maxDepth$}
				\State $maxDepth = depth$
				\State $maxPath = path$
			\EndIf
			\State \textbf{return} \label{l:terminate}
		\EndIf
		
		\For{$txn \in head.nextTxns$} \label{l:forloopdeepest}
			\State $path.push(txn)$
			\State \textsc{DeepestBranch}($txn, depth+1, path,$ \Statex \hspace{13em} $maxDepth,$ $maxPath$)
			\State $path.remove(txn)$
		\EndFor
		\EndProcedure
	\end{algorithmic}
\end{algorithm}

Line \ref{alg:serethSeries}:\ref{proc:series} refers to Series(), which iterates through each transaction in the list of Sereth transactions and forms graph relations between all transactions with corresponding mark/value hashes. Due to the uncertain nature of concurrency, it is possible for a transaction to have multiple potential successors, but only one predecessor. 

At line \ref{alg:serethSeries}:\ref{series:head} we locate from multiple potential head nodes the one that produces the deepest graph. From that graph, the deepest branch is our series. This logic mirrors that of the blockchain, in which branches are resolved by taking the longest branch.

\begin{figure}[ht]
\includegraphics[width=0.47\textwidth,keepaspectratio]{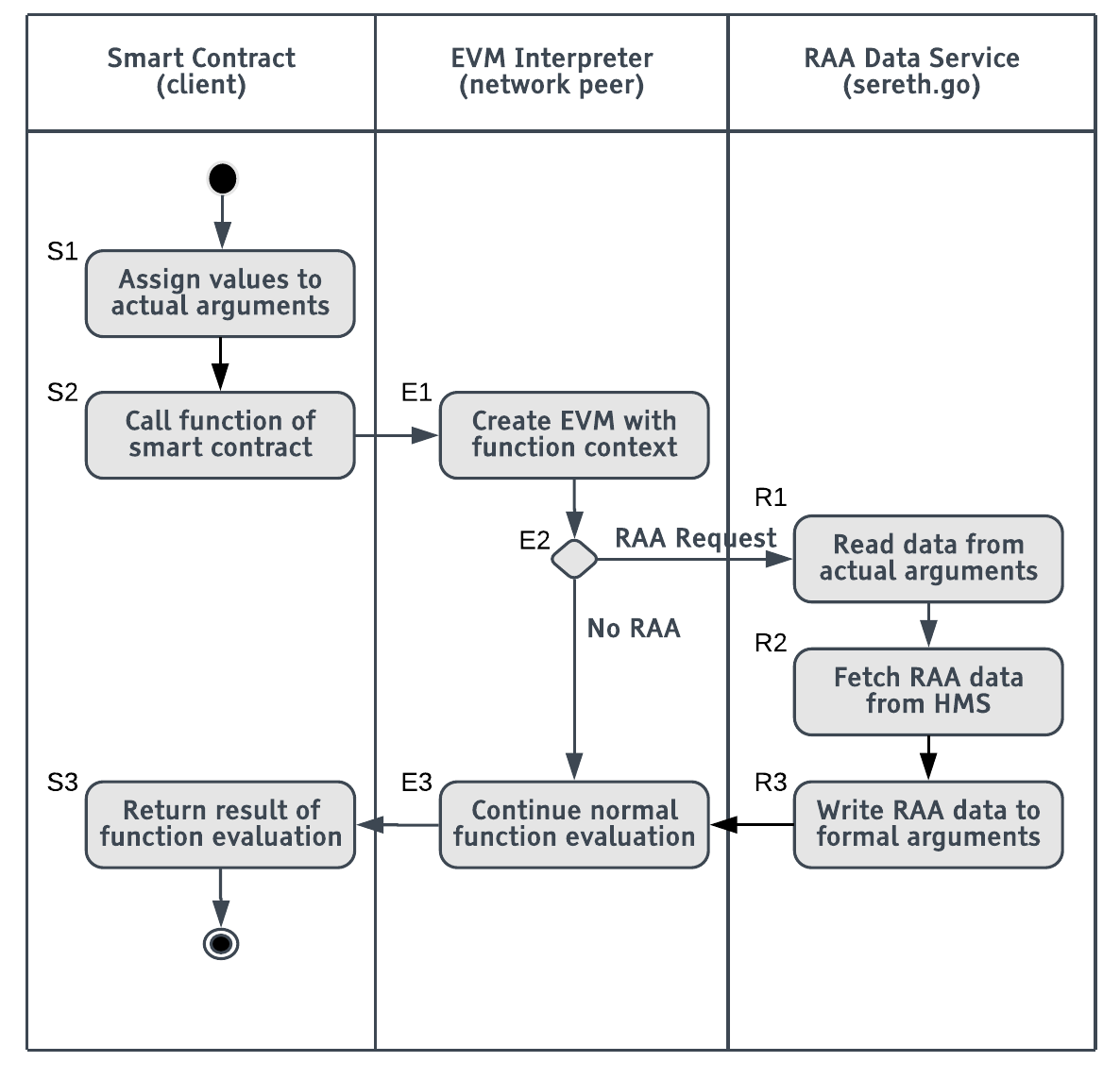}
\caption{RAA activity diagram}
\end{figure}

\subsection{Runtime Argument Augmentation}
\label{SubSection:RAA}
Blockchain oracles provide a secure and verifiable medium for smart contracts to access external data feeds, but still suffer from stale reads due to latency.
In our implementation of HMS it became clear that a traditional oracle would not satisfy the requirement for intra-block data. To overcome the limitations of oracles, we propose \textit{Runtime Argument Augmentation} (RAA), a technique that provides data to a smart contract by using the argument list as a channel to pass information.  RAA is a modification to the Ethereum Virtual Machine (EVM) interpreter, written in Golang.  Figure 1 is an activity diagram showing the modified processing.  In activity $E2$ the EVM interpreter checks to see if a function is requesting external data items using RAA.  If so, the interpreter calls the RAA provider in activities $R1$ to $R3$, implemented as a Golang service compiled into the EVM.  Data is obtained from the RAA provider and written into the function arguments.  The data types of the items being requested must match the data types of the arguments.  In $E3$ the function returns the result of evaluation using the modified arguments to the smart contract for use in activity $S3$.  RAA is flexible: any computation can be accomplished by the RAA provider, and the information can flow in both directions.  RAA is fast because it is written as an extension of the EVM.  A smart contract using RAA is indistinguishable to unmodified clients running Geth, who merely see that arguments are passed in and a result returned.

There are some limitations.  RAA cannot be used to modify the arguments of a smart contract function that may send a transaction.   This is because transactions along with their inputs are cryptographically signed by the sending address, stored in parameters \verb+msg.hash+ and \verb+msg.sender+.  Without this protection a malicious Geth client could modify the inputs of a transaction, for example doubling the price offered for an item or changing the delivery address.  In testing the limits of RAA we found that the modified transactions would still be mined, but would not be accepted by peers who must validate the newly created block.  In order to use RAA information in a transaction, a smart contract or other blockchain actor calls the RAA function first, then uses the information provided to improve the subsequent transaction.  This is the process used to obtain the experimental results that follow.

\section{Correctness}
\label{Correctness}

Concurrent systems are expected to satisfy correctness (safety) and progress (liveness) properties. Correctness is determined according to a defined correctness condition presented in literature~\cite{herlihy2011art}. HMS is designed for the sequential consistency correctness condition because miners are required to preserve the nonce order when committing a transaction from a given thread to a block. Since the nonce is a counter that reflects the sequential ordering of transactions issued by the same thread and a blockchain transaction is analogous to a concurrent method, the blockchain is inherently sequentially consistent. 
In the following lemma we show that HMS generates a series that provides a sequentially consistent ordering of the transactions in the longest branch. The benefit of generating a series of transactions in the longest branch is that it offers the greatest potential for optimum state throughput.

\begin{lemma}
The series generated from HMS preserves a sequentially consistent ordering of transactions invoked in the longest branch of the directed acyclic graph.
\end{lemma}
\begin{proof}
For each transaction $T$ in the transaction pool, if the signature is a set operation, and $T$ is either a possible head candidate or is a successor to the current tail in the series, then $T$'s mark is updated by hashing the predecessor transaction's mark and value, and the list of transactions to be considered for the series is amended to include transaction $T$. If the length of the list of transactions is larger than one, then HMS generates a series of transactions by calling the \textsc{Series} function with the transaction list as input. It now must be shown that the generated series is both sequentially consistent and the longest branch. The \textsc{Series} function creates an adjacency list of all transactions in the transaction list such that a transaction $T_2$ that is a member of $T_1$'s list indicates that $T_2$ is a successor to $T_1$. The \textsc{Series} function then iterates through the potential head candidates and applies the \textsc{DeepestBranch} algorithm. Each recursive call to \textsc{DeepestBranch} will iterate through the list of successor transactions in the adjacency list and apply \textsc{DeepestBranch} to each successor transaction until a transaction with no successors is reached. At each recursive call to \textsc{DeepestBranch}, transaction $txn$ passed as an input parameter is amended to the path. Since the exploration of the adjacency list guarantees that all successor transactions are visited after a predecessor transaction, any path generated from \textsc{DeepestBranch} will be sequentially consistent because the program order established in the adjacency list is preserved. Since the depth at each recursive call of \textsc{DeepestBranch} is incremented by one, and a path that exceeds the maximum depth is recorded upon termination of the recursive calls, the final recorded path by \textsc{DeepestBranch} will be the longest branch within the adjacency list.
\end{proof}

Progress of the underlying blockchain (the computer) is assumed.  We focus here on the progress of smart contract methods using a view of state variables managed by HMS.  \textit{Lock freedom} is defined as ensuring that some concurrent actor makes progress, and this is true for the blockchain as a whole but not for an individual smart contract.  Miners may assign a low priority to a particular contract so it makes no progress.  At peak times, many more transactions are sent to the network than can be included in a block.  Transactions sent may be lost due to network failures, memory limitations or peers not replaying them.  Miners may refuse to include transactions for arbitrary reasons.  
As a result, the progress guarantee provided by Ethereum is smart contract termination~\cite{amani2018towards,le2018proving}. Since the TxPool is a finite list of transactions, Algorithm~\ref{alg:serethParse} trivially terminates. Algorithm~\ref{alg:sereth} and Algorithm~\ref{alg:serethSeries} terminate given that the recursive function \textsc{DeepestBranch} terminates. We now show in the following lemma that \textsc{DeepestBranch} terminates.


\begin{lemma}
\textsc{DeepestBranch} presented in Algorithm~\ref{alg:serethSeries} is guaranteed to terminate.
\end{lemma}
\begin{proof}
The $txnList$ in the \textsc{Series} function is a finite list of transactions because it is a subset of the TxPool generated by the \textsc{Process} function. Therefore, each list within the adjacency list of transactions constructed by the nested for-loop on line~\ref{l:nestedforloop} of Algorithm~\ref{alg:serethSeries} will also contain a finite number of transactions. Since the $txn.mark$ computed by \textsc{Process} establishes an ordering among the transactions in $txnList$, the adjacency list of transactions will not contain any cycles due to the if-statement on line~\ref{l:ifstatementseries} of Algorithm~\ref{alg:serethSeries}. \textsc{DeepestBranch} will be invoked by the \textsc{Series} function no more than the number of transactions contained in $txnList$. For each invocation of \textsc{DeepestBranch}, a recursive call to \textsc{DeepestBranch} is made for each transaction in $head$'s list of successor transactions. Since \textsc{DeepestBranch} is only invoked on the successors of $head$, and each list in the adjacency list of transactions is finite, it is guaranteed that every invocation of \textsc{DeepestBranch} will eventually reach a transaction with no successors. Upon reaching a transaction with no successors, \textsc{DeepestBranch} terminates on line~\ref{l:terminate} of Algorithm~\ref{alg:serethSeries}.
\end{proof}

\section{Results}

This section shows experimental results of tests of the HMS algorithm on a private Ethereum blockchain.  The chain used for testing is a fork of an open source multi-peer private network configuration \cite{Chu2017-wv}.  Experiments were hosted on Ubuntu 16.04 EC2 servers in the AWS cloud.  The private network was configured to be a model of the Ethereum mainnet or the Ropsten testnet.  Proof of work was used as the consensus mechanism.  The block difficulty, transaction fees, processing power of the peers and peering topology were adjusted to produce block size and interval in the range of production Ethereum blockchains. 

Interoperability was tested by running experiments with a mix of peers running standard Geth and modified Sereth clients.  The first experiments were qualitative to demonstrate practical use of the two innovations of this paper: HMS and RAA.  Smart contract functions that created transactions were followed through the process of invocation, interpretation, transactions sent to the TxPool, replay, mining and validation.  The Sereth client operated interchangeably with Geth clients on the same network.  This is not surprising because Ethereum already supports a variety of clients with subtle differences, all following the same protocol.  Deployment of Sereth in the wild would not require a fork or any special permission from the network.  The Solidity smart contract equipped with RAA also functioned even when deployed to a Geth client, although of course the substitution of arguments did not take place and they were returned unchanged.  

Next we demonstrated that a sequential history was properly handled by sending a series of test transactions from the address of a single peer so that there is only one possible history, where real time order equals nonce order equals block order.  As expected, the transaction failure rate was zero and the transaction efficiency $\eta$ was 1.0.  
 
The quantitative experiments using concurrent peers demonstrated the effectiveness of HMS and the importance of transaction efficiency.  Experiments considered the history of program execution on a single shared variable P where P is an object containing the AMV tuple described in the HMS algorithm.  The dynamic pricing exchange from Section~\ref{SubSection:UseCases} is used to motivate the experiments, with the value of P representing the price.  Two transaction types are used in the experiments: $buy$ (buys one item at the current price) and $set$ (changes the price).  A ratio of buys (\textsc{Read-Uncommitted}) to sets (\textsc{Writes}) was used as a non-dimensional parameter that would scale up to larger servers as the absolute number of transactions increased.  The number of set transactions was varied from 100 to 5, yielding a buy to set ratio of from 1:1 to 20:1.  

Figure 2 depicts a plot of state throughput measured at different buy set ratios.  Each data point represents the result of 100 buy transactions, so state throughput is equivalent to $\eta$ expressed as a percentage. Transactions were submitted at an interval of one second, resembling a moderate throughput smart contract use case.  This interval was sufficient to demonstrate the problem of stale reads and can easily be reproduced with ordinary servers using the provided source code.  The sets are evenly spaced over the processing of the buys.  
The lines are smoothed averages of the points shown, with the shaded areas representing the 90 percent confidence interval for the lines.    
  
\subsection{Standard Geth client}

The baseline scenario sends transactions to an unmodified peer running the standard Geth client.  The transaction efficiency at different buy to set ratios is labeled as ``geth\_unmodified'' in Figure 2. In this scenario, buy transactions that read the price P from block $n-1$ and are included in block $n$ before the price is modified will be successful, while all other buys will fail.  When there are many price sets, as in the experiments with 1:1 and 2:1 ratios, only a few buys are successful. In some runs no buy transactions succeeded at all.  The efficiency increases somewhat as the ratio of buys to sets goes above 10:1 because there are more buys reading correct values before an intra-block set occurs.  However it remains poor for two reasons. 

First, with a low ratio of buys it is unlikely for a buy to land in the very beginning of the block before any sets take place.  Thus many will fail.  Second, even as the ratio increases, because of the large transaction pool there are often no buys going into block $n+1$ that have a valid view of block $n$.  Instead, block $n$ is assembled from buys that were submitted a few blocks ago, so they may have a view of block $n-2$,  $n-3$ or older blocks.  These buys fail because the blockchain state has passed them by before they were included.  This phenomenon is frequently observed in public blockchains \cite{Easley2018-yb}.

Although not shown in the plot, it was also observed that with few state changes (high ratio of buys) transaction efficiency becomes more sensitive to the transaction interval, as miners may sequence a large number of buys together.   

Sets are not plotted.  All of the sets succeed because they are sent from the owner of the contract and they do not depend on the previous price.  If sets came from different addresses some might have failed, but it is reasonable that the owner of the contract is the only one allowed to set the price.

\begin{figure}[t!]
    \centering
    \includegraphics[width=0.47\textwidth]{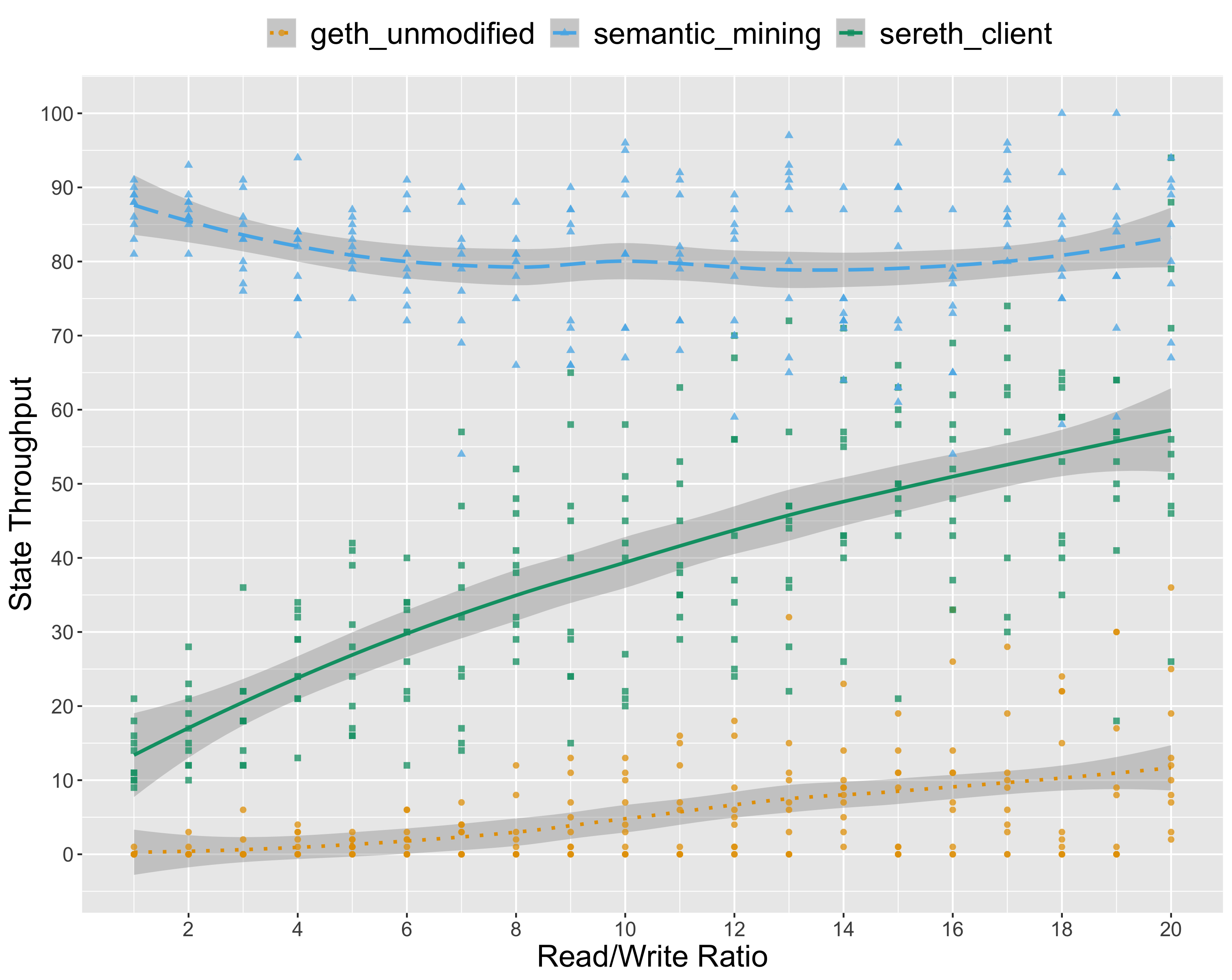}
    \caption{Transaction efficiency $\eta$ vs \textsc{Read-Uncommitted / Write} ratio}
    \label{Fig:TxnEfficiency}
\end{figure}

\subsection{Hash-Mark-Set without miner assistance}

The second experimental scenario, labeled as ``sereth\_client'' in Figure 2, used the modified Sereth client on the network, implementing the HMS algorithm.  The set transactions were ordered with HMS while buy transactions were sent exactly as in the baseline scenario.  Interleaved with the sets, any buy at the right mark and price succeeded.  The benefit of HMS in this scenario is that the buy transactions have a \textsc{Read-Uncommitted} view of the likely state of the storage variable P when they will be evaluated.  This allows many more transactions to succeed.  A sequentially consistent ordering of the set operations was established and their dependent buys have a view of the state provided by HMS.
Figure~\ref{Fig:TxnEfficiency} shows an improvement in throughput by approximately a factor of five over the entire range of read /write ratios.  These results were achieved without miner assistance, so they could be accomplished simply by running the modified client on the public Ethereum blockchain, as long as access to the smart contract was via these clients.

This experiment also demonstrates how HMS alleviates the intra-block lost update problem.  The FPV arguments in each buy include the previous mark, a hash that relates it to an interval between two sets.  If a sequence occurs such as: set(5), buy(5), set(7), set(5), buy(5), a particular buy(5) can prove that it was sent during the first or the second interval the price was set to 5.   Linking each buy transactions to a particular set price prevents the frontrunning attack mentioned in Section~\ref{SubSection:UseCases}. 

\subsection{Hash-Mark-Set with semantic mining}

In the third experimental scenario, the inputs of the second scenario were repeated with the miner using the HMS algorithm to determine the block order of transactions.  In this scenario HMS information about the TxPool is available to both smart contract users and miners.  Since the miner now has awareness of the semantics of the transactions, we call this \textit{semantic mining}.  Previously miners would not reorder transactions to increase transaction efficiency, but the semantic miners have this capability.  The line labeled ``semantic\_mining'' in Figure~\ref{Fig:TxnEfficiency} shows the results.  About 80 percent of transactions succeed due to semantic mining providing interleaving in conformance to the \textsc{Read-Uncommitted} view used by the smart contract clients when they sent the transactions.   Relative improvement in throughput was greatest with a high frequency of price changes, i.e. where there are 1 or 2 buys per set.  At these ratios the advantage of having the miner interleave transactions increases transaction efficiency from a few percent to almost 90 percent, resulting in a factor of six over the unassisted case.  Overall, 10-20 percent of transactions were lost due to the fact that the TxPool no longer contains marked transactions immediately after the block is published.  Transaction efficiency could approach 100 percent if HMS were extended to include the final values from replaying each block.  Other factors that would impact efficiency is if only a fraction of the miners were assisting, or if communication of the TxPool were impeded among the Sereth enabled peers.  Performance would be degraded in these cases but there would still be benefits proportional to the participation.

\section{Related Work}



There are five main approaches to improve blockchain throughput and latency: Reparameterization, sidechains, sharding, leader election and invalid state tagging.  

Reparameterization involves tuning the block size and interval to network bandwidth and peer computing power \cite{Croman2016-mg}.  HMS does not use reparameterization, but could influence tuning trade-offs by decreasing the significance of block interval. 

Sidechains \cite{Back2014-qs} increase throughput by creating transaction channel networks such as Lightning \cite{Poon2015-lm}.  As the name implies they exploit parallelism by running multiple chains, merging them to the main chain as needed to ensure correctness.  Sidechains have been implemented at scale on existing blockchains.   Recently generalized formally as \textit{state channels} \cite{Coleman2018-aj}, they can provide throughput gains of several orders of magnitude.  However the authors note state channels do not solve the latency, or as they call it, ``time granularity'' problem.  The \textsc{Read-Uncommitted} view provided by HMS does solve this for specific state variables.  

Sharding \cite{Luu2016-ux} increases throughput by isolating segments of the blockchain peer network.  Sharding requires changes to consensus protocol but has been accepted by Ethereum \cite{Buterin_Vitalik_2018-nk} with significant progress \cite{Prysmatic_Labs2018-du} and a target implementation date in 2020.  Like state channels, sharding is inherently parallel and offers performance gains of several orders of magnitude.  Sharding is a global solution but would need customization to address state throughput of individual smart contracts as does HMS.   

Bitcoin-NG ``Next Generation'' \cite{Eyal2016-fr} uses leader election with continuous serialization to modify the consensus protocol in proof of work blockchains such as Bitcoin.   Performance gains scaling to the limits of network latency and individual peer processing power are reported.  Recent work \cite{Yin2018-mb} notes the history of improvement and elaboration on the original proposal.  Our solution shares with Bitcoin-NG the concept of continuous serialization to reduce the ``long freeze'' of latency, but HMS does not require protocol changes to interoperate with current blockchains. 

The scope of our review was public blockchains, however a Byzantine Fault Tolerant (BFT) proposal for Hyperledger is relevant because it focused on the bottleneck of transaction signing and ordering in block creation \cite{Sousa2018-ps}.   BFT uses a leader to coordinate block creation achieving transaction rates of over 2000 per second on private blockchains.  Unlike Ethereum, Hyperledger tags transactions known to be based on invalid states before they are ordered in a block, so time is not wasted replaying the failed transactions.   However the authors do not consider transaction efficiency and the BFT ordering service does not use semantics to reduce failures as HMS does.

Software Transactional Memory (STM) algorithms have also been applied to blockchain throughput.  The original ideas in \cite{Dickerson2017-rn} using STM to enable concurrent processing of smart contract methods were continued by \cite{Anjana2018-nb}.  These researchers note the EVM is not parallel and the difficulty of determining transaction dependencies in a block, so in both papers smart contracts were translated into C++ which is supported by the STM library.  Simulated miners then interleave and order smart contract methods in STM to create a concurrent execution.  Speedups of up to 2x were achieved.   Concurrency based throughput gains in which ``any sequential execution will do'' are different from HMS, which sequences smart contract methods for transaction efficiency and increased state throughput. 

A parallel to HMS is found in an earlier STM paper \cite{Spear2008-nm} whose language about ``publishing'' is prescient as it was written before the blockchain was invented.  A correctness condition called Selective Strict Serialization (SSS) is introduced, in which some transactions are strictly serialized and others are not, but are marked to the serialized history.  In Section~\ref{Correctness} above we applied sequential consistency as the correctness condition for our HMS algorithm.  In our experiments HMS establishes a fixed ordering for the state changes (sets) while allowing the dependent transactions (buys) to be arbitrarily interleaved. Multiple buys can occur in a price interval and are not dependent upon each other.  Within the interval any order of buys is valid so they do not require an established ordering constraint.  Further work might show that SSS is a correctness condition suitable for HMS.

There is relevant work on improving throughput and latency of concurrent systems by reducing \textit{abort rate}, defined as how many times a transaction is retried before success \cite{Wu2016-yy, Chen2017-wl, Yuan2016-hx, Cohen2018-vb}.  This is different from our \textit{state throughout}, which measures efficiency of blockchain commits that are not repeated.  
High abort rates due to delayed write visibility, where transaction writes may only be read after commit, is addressed by Faleiro et al.~\cite{faleiro2017high} in the proposal of piece-wise visibility (PWV), a deterministic concurrency control protocol designed to enable early write visibility. 
PWV divides a transaction into a set of sub-transactions which are scheduled to be executed in a serializable order.
Each sub-transaction write is made visible as soon as it commits, enabling the original transaction writes to be visible prior to commit time.
A DAG is used to order database sub-transactions based on data dependencies.
HMS uses a DAG to order blockchain transactions in a sequentially consistent fashion, and the final series of transactions is derived from the deepest branch.

The fundamental difference between PWV and HMS is that PWV enables writes to be visible inside the commit protocol while HMS enables 
\textsc{Read-Uncommitted} isolation 
for smart contracts through our proposed RAA technique, described in Section~\ref{SubSection:RAA}.   
The PWV commit protocol only provides write visibility after a transaction is submitted to the database system, which limits the potential performance gains in comparison to HMS that provides write visibility to smart contract clients such that the requested data from RAA can be utilized prior to transaction submission.


\section{Conclusions}

State throughput, the throughput of successful transactions, is proposed as the appropriate metric for smart contract performance.  An algorithm, Hash-Mark-Set, and a novel architectural technique, Runtime Argument Augmentation, are presented and demonstrated together on the Ethereum blockchain to improve state throughput.

HMS provides smart contracts a \textsc{Read-Uncommitted} view of state.  At the same time, HMS provides information about transaction dependencies to the miners so they can adjust the block order, called semantic mining. Miners cooperating with smart contracts using the HMS algorithm to order dependent transactions were able to create blocks in which most transactions were successful.  This is demonstrated to improve transaction efficiency from less than 5 percent to over 80 percent in cases where state changes are frequent, more than an order of magnitude improvement in state throughput.  Even without semantic mining, the \textsc{Read-Uncommitted} view is helpful, increasing state throughput by a factor of five across the full range of tested read to write ratios from 1:1 to 20:1.  
Latency (as a function of correct reads) was also reduced in both scenarios, client modifications only  and semantic mining.  In addition to the performance gains, HMS solves the blockchain lost update and frontrunning attack problems because transactions using \textsc{Read-Uncommitted} values keep a unique hash validated record of the particular interval during which the value was read.  

RAA is presented as a new technique to provide smart contracts rapid communication with external data services.  In our experiments smart contracts used RAA to access \textsc{Read-Uncommitted} views of data necessary for transaction success and thus increase transaction efficiency.  RAA works at the architectural level of the EVM, using the interpreter to achieve high performance.  Peers running the RAA modified client were demonstrated to work interoperably with standard peers. 

\section*{Acknowledgment}
  Thanks to Raul Jordan and Andreas Olofsson for assistance in understanding the intricacies of Geth client software.

\balance


\end{document}